\documentclass[twoside,11pt]{article}

\usepackage{jmlr2e}
\usepackage{algorithm}
\usepackage{algorithmic}
\usepackage{subfig}

\firstpageno{1}

\newcommand{\OPT}{\mathrm{OPT}}
\newcommand{\dist}{\mathrm{dist}}
\newcommand{\pr}{\mathrm{pr}}
\newcommand{\iter}{\mathrm{iter}}
\newcommand{\Clustered}{\mathrm{Clustered}}
\newcommand{\Components}{\mathrm{Components}}
\newcommand{\Comp}{\mathrm{Comp}}
\begin{document}

\title{Efficient Clustering with Limited Distance Information\thanks{A preliminary version of this article appeared under the same title in the Proceedings of the Twenty-Sixth Conference on Uncertainty in Artificial Intelligence, AUAI Press, Corvallis, Oregon, 632-641.}}

\author{\name Konstantin Voevodski \email kvodski@bu.edu \\
       \addr Department of Computer Science\\
       Boston University\\
       Boston, MA 02215, USA
       \AND
       \name  Maria-Florina Balcan \email ninamf@cc.gatech.edu \\
       \addr College of Computing\\
       Georgia Institute of Technology\\
       Atlanta, GA 30332, USA
		\AND
       \name Heiko R{\"o}glin \email heiko@roeglin.org \\
       \addr Department of Quantitative Economics\\
       Maastricht University\\
       Maastricht, The Netherlands
		\AND
       \name Shang-Hua Teng \email shanghua@usc.edu \\
       \addr Computer Science Department\\
       University of Southern California\\
       Los Angeles, CA 90089, USA
		\AND
       \name Yu Xia \email yuxia@bu.edu \\
       \addr Bioinformatics Program and Department of Chemistry \\
       Boston University\\
       Boston, MA 02215, USA}

\editor{}

\maketitle

\begin{abstract}
Given a point set $S$ and an unknown metric $d$ on $S$, we study the problem
of efficiently partitioning $S$ into $k$ clusters while querying few distances
between the points.  In our model we assume that we have access to \emph{one
versus all} queries that given a point $s \in S$ return the distances between $s$
and all other points.  We show that given a natural assumption about the
structure of the instance, we can efficiently
find an accurate clustering using only $O(k)$ distance queries.   Our algorithm uses an
\emph{active} selection strategy to choose a small set of points that we call landmarks,
and considers only the distances between landmarks and other points to produce a clustering.
We use our procedure to cluster proteins by sequence similarity.  This setting
nicely fits our model because we can use a fast sequence database search program to
query a sequence against an entire dataset.  We conduct an empirical study that
shows that even though we query a small fraction of the distances between the
points, we produce clusterings that are close to a desired clustering given by
manual classification.
\end{abstract}


\section{Introduction}

Clustering  from pairwise distance information  is an important problem in the
analysis and exploration of data. It has many variants and formulations and it
has been extensively studied in many different communities, and many different
clustering algorithms have been proposed.

Many application domains ranging from computer vision to biology have recently
faced an explosion of data, presenting several challenges to traditional
clustering techniques.  In particular, computing the distances between all pairs
of points, as required by traditional clustering algorithms, has become
infeasible in many application domains. As a consequence it has become
increasingly important to develop effective clustering algorithms that can
operate with limited distance information.

In this work we initiate a study of clustering with limited distance information;
in particular we consider clustering with a small number of \emph{one versus all}
queries. We can imagine at least two different ways to query distances between
points. One way is to ask for distances between pairs of points, and the other is
to ask for distances between one point and all other points.  Clearly, a one
versus all query can be implemented as $|S|$ pairwise queries, but we draw a
distinction between the two because the former is often significantly faster in
practice if the query is implemented as a database search.

Our main motivating example for considering one versus all distance queries is
sequence similarity search in biology.  A program such as BLAST~\citep{blast}
(Basic Local Alignment Search Tool) is optimized to search a single sequence
against an entire database of sequences.  On the other hand, performing $|S|$
pairwise sequence alignments takes several orders of magnitude more time, even if
the pairwise alignment is very fast.  The disparity in runtime is due to the
hashing that BLAST uses to identify regions of similarity between the input
sequence and sequences in the database.  The program maintains a hash table of
all \emph{words} in the database (substrings of a certain length), linking each
word to its locations.  When a query is performed, BLAST considers each
word in the input sequence, and runs a local sequence alignment in each of its
locations in the database.  Therefore the program only performs a limited number
of local sequence alignments, rather than aligning the input sequence to each
sequence in the database.  Of course, the downside is that we
never consider alignments between sequences that do not share a word.  However,
in this case an alignment may not be relevant anyway, and we can assign a
distance of infinity to the two sequences.  Even though the search performed by BLAST
is heuristic, it has been shown that protein sequence similarity identified by BLAST is meaningful~\citep{reliableBlast}.



Motivated by such scenarios, in this paper we consider the problem of clustering a dataset with an unknown
distance function, given only the capability to ask one versus all distance
queries.  We design an efficient algorithm for clustering accurately with a small number of such queries.  To formally analyze the correctness of our algorithm we
assume that the distance function is a metric, and that our clustering problem
satisfies a natural approximation stability property regarding the utility of the $k$-median objective function
in clustering the points. In particular, our analysis assumes the $(c,\epsilon)$-property
of \citet{bbg}.  For an objective function $\Phi$ (such as k-median), the $(c,\epsilon)$-property
assumes that any clustering that is a $c$-approximation of $\Phi$ has error of at most
$\epsilon$. To define what we mean by error we assume that there exists some unknown
relevant ``target'' clustering $C_{T}$; the error of a proposed clustering $C$ is then the fraction
of misclassified points under the optimal matching between the clusters in  $C_{T}$ and $C$.
%

Our first main contribution is designing an algorithm that given the $(c,\epsilon)$-property for the k-median objective finds an accurate clustering with probability at least $1-\delta$ by using only $O(k + \ln \frac{1}{\delta})$ one versus all queries.  In particular, we use the same assumption as
\citet{bbg}, and we obtain effectively the same performance guarantees
as~\citet{bbg} but by only using a very small number of one versus all queries.
In addition to handling this  more difficult scenario, we also
provide a much faster algorithm.  The algorithm of~\citet{bbg} can be
implemented in $O(|S|^{3})$ time, while the one proposed here runs in time
$O((k + \ln \frac{1}{\delta})|S| \log |S|)$.

Our algorithm uses an \emph{active} selection strategy to choose a small set of landmark points.  We then construct a clustering using only the distances between landmarks and other points.  The runtime of our algorithm is $O(\vert L \vert |S| \log |S|)$, where $L$ is the set of landmarks that have been selected.  Our adaptive selection procedure significantly reduces the query and time complexity of the algorithm.  We show that using our adaptive procedure it suffices to choose only $O(k + \ln \frac{1}{\delta})$ landmarks to produce an accurate clustering with probability at least $1-\delta$.  Using a random selection strategy we need $O(k \ln \frac{k}{\delta})$ landmarks if the clusters of the target clustering are balanced in size, and otherwise performance degrades significantly because we may need to sample points from much smaller clusters.

We use our algorithm to cluster proteins by sequence similarity, and compare our results to gold standard manual classifications given in the Pfam \citep{pfam} and SCOP \citep{scop} databases.  These classification databases are used ubiquitously in biology to observe evolutionary relationships between proteins and to find close relatives of particular proteins.  We find that for one of these sources we obtain clusterings that usually closely match the given classification, and for the other the performance of our algorithm is comparable to that of the best known algorithms using the full distance matrix.   Both of these classification databases have limited coverage, so a completely automated method such as ours can be useful in clustering proteins that have yet to be classified.  Moreover, our method can cluster very large datasets because it is efficient and does not require the full distance matrix as input, which may be infeasible to obtain for a very large dataset.

\smallskip
\noindent{\bf Related Work:~~}
A property that is related to the $(c,\epsilon)$-property is $\epsilon$-separability, which was
introduced by \citet{orss}.  A clustering instance is
$\epsilon$-separated if the cost of the optimal $k$-clustering is at most
$\epsilon^{2}$ times the cost of the optimal clustering using $k-1$ clusters.
 The $\epsilon$-separability and
$(c,\epsilon)$ properties are related: in the case when the clusters are large the \citet{orss} condition implies the \citet{bbg} condition \citep[see][]{bbg}.



Ostrovsky et al. also present a sampling method for choosing initial centers, which when followed by a single Lloyd-type descent step gives a constant factor approximation of the $k$-means objective if the instance is $\epsilon$-separated.  However, their sampling method needs
information about the full distance matrix because the probability of picking
two points as two cluster centers is proportional to their
squared distance. A very similar (independently
proposed) strategy is used by \citet{av} to obtain an $O(\log{k})$-approximation of
the $k$-means objective on arbitrary instances. Their work
was further extended by \citet{ajm} to give a constant factor approximation using
$O(k \log k)$ centers. The latter two algorithms can be implemented with $k$ and $O(k \log k)$ one versus all distance queries, respectively.

\citet{abs} have since improved the approximation guarantee of \citet{orss} and some of the results of \citet{bbg}. In particular, they show a way
to arbitrarily closely approximate the $k$-median and $k$-means objective when the \citet{bbg} condition is satisfied and all the target clusters are large. In their analysis they use a property called weak deletion-stability, which is implied by the \citet{orss}
condition and the \citet{bbg}  condition when the target clusters are large. However, in order to find a $c$-approximation (and given our assumption a clustering that
is $\epsilon$-close to the target) the runtime of their algorithm is $n^{O(1/(c-1)^{2})}k^{O(1/(c-1))}$. On the
other hand, the runtime of our algorithm is completely independent of $c$, so it remains
efficient even when the $(c,\epsilon)$-property holds only for some very small constant $c$.

Approximate clustering using sampling has been studied extensively in recent
years
\citep[see][]{sampleBasedClustering,sampleBasedClustering2,sampleBasedClustering3}.  The methods proposed in these papers yield constant factor approximations to the
$k$-median objective
using at least $O(k)$ one versus all distance queries. However, as the constant
factor of these approximations is at least $2$, the proposed sampling methods
do not necessarily yield clusterings close to the target clustering $C_T$
if the $(c,\epsilon)$-property holds only for some small constant $c<2$, which
is the interesting case in our setting.

Our landmark selection strategy is related to the \emph{farthest first traversal} used by \citet{dasgupta}.  In each iteration this traversal selects the point that is farthest from the ones chosen so far, where distance from a point $s$ to a set $X$ is given by min$_{x \in X} d(s,x)$.  This traversal was originally used by \citet{gonzalez} to give a 2-approximation to the $k$-center problem.  It is used in \citet{dasgupta} to produce a hierarchical clustering where for each $k$ the induced $k$-clustering is a constant factor approximation of the optimal $k$-center clustering.  Our selection strategy is somewhat different from farthest first traversal because in each iteration we uniformly at random choose one of the furthest points from the ones selected so far.  In addition, the theoretical guarantees we provide are quite different from those of Gonzales and Dasgupta.

\section{Preliminaries}

Given a metric space $M=(X,d)$ with point set $X$, an unknown distance function
$d$ satisfying the triangle inequality, and a set of points $S \subseteq X$, we
would like to find a $k$-clustering $C$ that partitions the points in $S$ into
$k$ sets $C_{1},\ldots,C_{k}$ by using \emph{one versus all} distance
queries.

In our analysis we assume that $S$ satisfies the $(c,\epsilon)$-property of
\citet{bbg} for the $k$-median objective function.  The
$k$-median objective is to minimize $\Phi(C) = \sum_{i=1}^{k}\sum_{x \in C_{i}}
d(x,c_{i})$, where $c_{i}$ is the median of cluster $C_{i}$, which is the point $y \in C_{i}$ that minimizes $\sum_{x \in C_{i}}d(x,y)$. Let $\OPT_{\Phi}
= \min_{C}\Phi(C)$, where the minimum is over all $k$-clusterings of $S$,
and denote by  $C^{\ast} = \lbrace C_{1}^{\ast},\ldots,
C_{k}^{\ast} \rbrace$ a clustering achieving this value.

To formalize the $(c,\epsilon)$-property we need to define a notion of distance
between two $k$-clusterings $C = \lbrace C_{1},\ldots,C_{k} \rbrace$ and
$C' = \lbrace C'_{1},\ldots,C'_{k} \rbrace$. As in \citep{bbg}, we define the
distance between $C$ and $C'$ as the fraction of points on which they disagree
under the optimal matching of clusters in $C$ to clusters in $C'$:
\begin{displaymath}
   \dist(C,C') = \min_{\sigma \in S_{k}} \frac{1}{n} \sum_{i=1}^{k} \vert C_{i} - C_{\sigma(i)}' \vert,
\end{displaymath}
where $S_{k}$ is the set of bijections $\sigma\colon \lbrace 1,\ldots,k  \rbrace \rightarrow \lbrace 1,\ldots,k  \rbrace$.  Two clusterings $C$ and $C'$ are \emph{$\epsilon$-close}
if $\dist(C,C') < \epsilon$.

We assume that there exists some unknown relevant ``target'' clustering $C_{T}$ and given a proposed clustering $C$ we define the error of $C$ with respect to $C_{T}$ as $\dist(C,C_{T})$. Our goal is to find a clustering of low error.

The $(c,\epsilon)$-property  is defined as follows.
\begin{definition}
We say that
 the instance $(S,d)$ satisfies the $(c,\epsilon)$-property for the $k$-median objective function
with respect to the target clustering $C_T$ if any
clustering of $S$ that approximates $\OPT_{\Phi}$ within a factor of $c$ is
$\epsilon$-close to $C_{T}$, that is,
$
   \Phi(C) \le c \cdot \OPT_{\Phi} \Rightarrow \dist(C,C_{T}) < \epsilon.
$
\end{definition}

In the analysis of the next section we denote by $c_{i}^{\ast}$ the center point
of $C_{i}^{\ast}$, and use $\OPT$ to refer to the value of $C^{\ast}$ using the
$k$-median objective, that is, $\OPT = \Phi(C^{\ast})$. We define the
\emph{weight} of point $x$ to be the contribution of $x$ to the $k$-median
objective in $C^{\ast}$:  $w(x) = \min_{i}d(x,c_{i}^{\ast})$.  Similarly, we use
$w_{2}(x)$ to denote $x$'s distance to the second-closest cluster center among
$\lbrace c_{1}^{\ast},c_{2}^{\ast},\ldots,c_{k}^{\ast} \rbrace$. In addition, let
$w$ be the average weight of the points:
$
   w = \frac{1}{n} \sum_{x \in S} w(x) = \frac{\OPT}{n},
$
where $n$ is the cardinality of $S$.

\section{Clustering With Limited Distance Information}

\begin{algorithm}[H]
\caption{Landmark-Clustering($S,\alpha,\epsilon,\delta,k$)}
\begin{algorithmic}
\STATE $b = (1 + 17/ \alpha) \epsilon n$;
\STATE $q=2b$;
\STATE $\iter=4k + 16 \ln \frac{1}{\delta}$;
\STATE $s_{\min} = b+1$;
\STATE $n' = n-b$;
\STATE $L$ = \textbf{Landmark-Selection}$(q,\iter)$;
\STATE $C'$ =  \textbf{Expand-Landmarks} $(s_{\min},n',L)$;
\STATE Choose some landmark $l_{i}$ from each cluster $C'_{i}$;
\FOR{each $x \in S$}
\STATE Insert $x$ into the cluster $C''_{j}$ for $j$ = $\mathrm{argmin}_{i}d(x,l_{i})$;
\ENDFOR
\RETURN $C''$;
\end{algorithmic}
\label{alg-main}
\end{algorithm}

In this section we present a new algorithm that accurately clusters a set of
points assuming that the clustering instance satisfies the $(c,\epsilon)$-property for $c=1+\alpha$,
and the clusters in the target clustering $C_T$ are not too small.
The algorithm presented here is much faster than the one
given by Balcan et al., and does not require all pairwise distances as input.
Instead, we only require $O(k + \ln \frac{1}{\delta})$ one versus all distance queries to
achieve the same performance guarantee as in \citep{bbg} with probability at least $1-\delta$.

Our clustering method is described in Algorithm~\ref{alg-main}.
We start by using the \emph{Landmark-Selection} procedure to \emph{adaptively} select a small set of landmarks. This procedure repeatedly chooses uniformly at random one of the $q$ furthest points from the ones selected so far, for an appropriate $q$.  We use $d_{min}(s)$ to refer to the minimum distance between $s$ and any point selected so far.  Each time we select a new landmark $l$, we use a one versus all distance query to get the distances between $l$ and all other points in the dataset, and update $d_{min}(s)$ for each point $s \in S$.  To select a new landmark in each iteration, we choose a random number $i \in \lbrace n-q+1,\ldots,n \rbrace$ and use a linear time selection algorithm to select the $i$th furthest point.  We note that our algorithm only uses the distances between landmarks and other points to produce a clustering.


\begin{algorithm}[H]
\caption{Landmark-Selection($q,\iter$)}
\begin{algorithmic}
\STATE Choose $l \in S$ uniformly at random;
\STATE $L = \lbrace l \rbrace$;
\FOR{each $d(l,s) \in$ QUERY-ONE-VS-ALL$(l,S)$}
\STATE $d_{min}(s) = d(l,s)$;
\ENDFOR
\FOR{$i$ = 1 to $\iter - 1$}
\STATE Let $s_{1},...,s_{n}$ be an ordering of the points in $S$ such that $d_{min}(s_{i}) \le d_{min}(s_{i+1})$ for $i \in \lbrace 1,\ldots,n-1 \rbrace$;
\STATE Choose $l \in \lbrace s_{n - q + 1},\ldots,s_{n} \rbrace$ uniformly at random;
\STATE $L = L \cup \lbrace l \rbrace$;
\FOR{each $d(l,s) \in$ QUERY-ONE-VS-ALL$(l,S)$}
\IF{$d(l,s) <  d_{min}(s)$}
\STATE $d_{min}(s) = d(l,s)$;
\ENDIF
\ENDFOR
\ENDFOR
\RETURN $L$;
\end{algorithmic}
\end{algorithm}


\emph{Expand-Landmarks} then expands a ball $B_{l}$ around each landmark $l \in
L$ chosen by \emph{Landmark-Selection}. We use the variable $r$ to denote the
radius of all the balls: $B_{l} = \{s \in S \mid d(s,l) \le r \}$. The algorithm
starts with $r=0$, and increments it until the balls satisfy a property
described below.  For each $B_{l}$ there are $n$ relevant values of $r$ to try,
each adding one more point to $B_{l}$, which results in at most $\vert L
\vert n$ values to try in total.

The algorithm maintains a graph $G_{B} = (V_{B},E_{B})$, where vertices
correspond to balls that have at least $s_{\min}$ points in them, and two
vertices are connected by an (undirected) edge if the corresponding balls overlap
on any point: $(v_{l_{1}},v_{l_{2}}) \in E_{B}$ iff $B_{l_{1}} \cap B_{l_{2}} \ne
\emptyset$. In addition, we maintain the set of points in these balls $\Clustered
= \{ s \in S \mid \exists l\colon s \in B_{l} \}$ and a list of the connected
components of $G_{B}$, which we refer to as $\Components(G_{B}) = \lbrace
\Comp_{1},...,\Comp_{m} \rbrace$.

In each iteration, after we expand one of the balls by a point, we update
$G_{B}, \Components(G_{B})$, and $\Clustered$.  If $G_{B}$ has exactly
$k$ components, and $\vert \Clustered \vert \ge n'$, we terminate and report
points in balls that are part of the same component in $G_{B}$ as distinct
clusters.  If this condition is never satisfied, we report \textbf{no-cluster}.
A sketch of the algorithm is given below.  We use $(l^{\ast},s^{\ast})$ to refer to the next landmark-point pair that is considered, corresponding to expanding $B_{l^{\ast}}$ to include $s^{\ast}$ (Figure~\ref{fig:figure1}).

\begin{figure}
\begin{center}
\includegraphics[width=70mm,height=50mm]{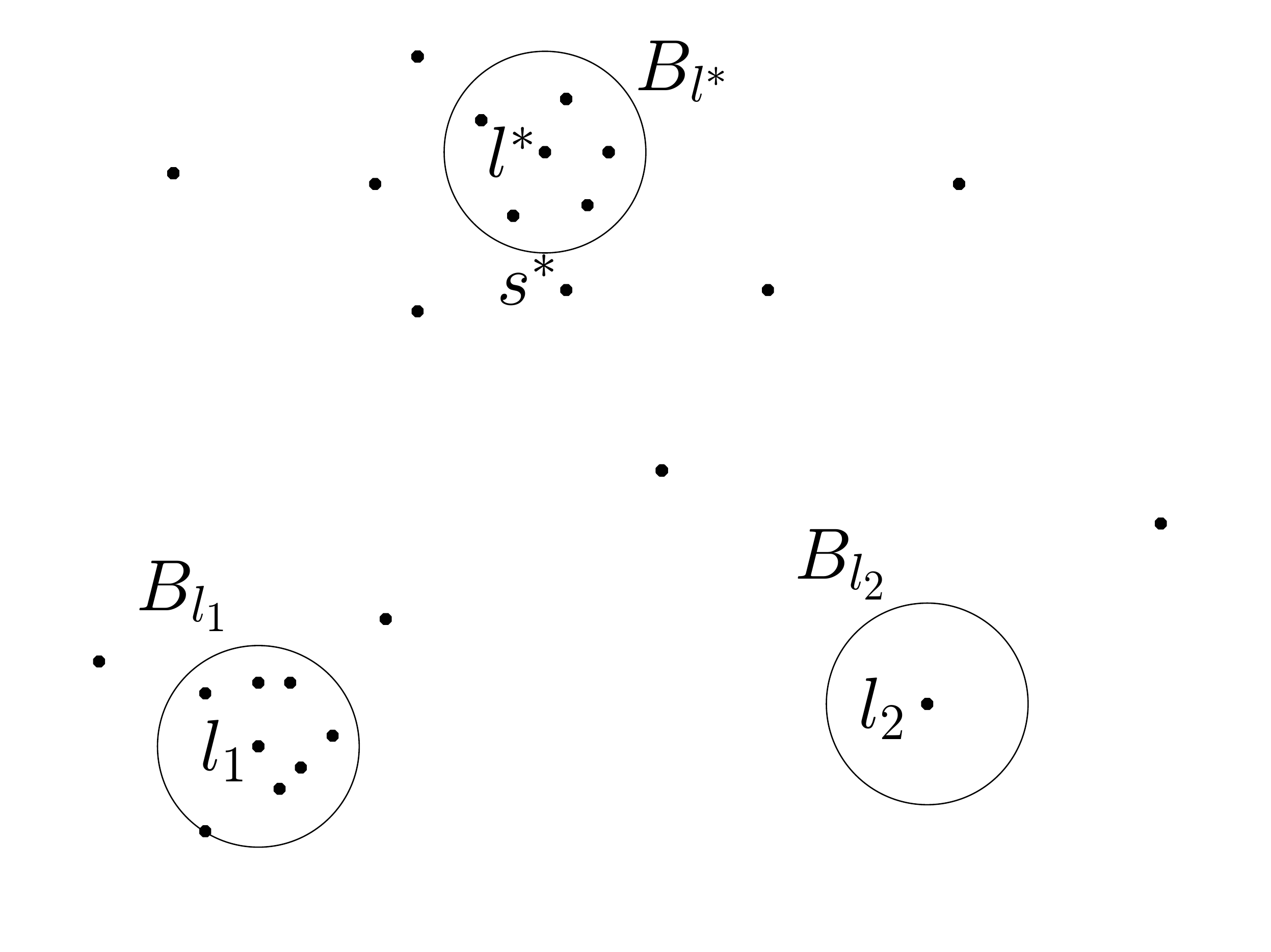}
\caption{Balls around landmarks are displayed, with the next point to be added to a ball labeled as $s^{\ast}$. \label{fig:figure1}}
\end{center}
\end{figure}

\begin{algorithm}[H]
\caption{Expand-Landmarks($s_{\min},n',L$)}
\begin{algorithmic}[1]
\WHILE{($(l^{\ast},s^{\ast})$ = Expand-Ball()) != null}
\STATE $r = d(l^{\ast},s^{\ast})$;
\STATE update $G_{B}$, $\Components(G_{B})$, and $\Clustered$
\IF{$\vert \Components(G_{B}) \vert = k$ and $\vert \Clustered \vert \ge n'$}
\RETURN $C = \lbrace C_{1},...,C_{k} \rbrace $ where $C_{i} = \lbrace s \in S \mid \exists l\colon s \in B_{l}$ and $v_{l} \in \Comp_i \rbrace$.
\ENDIF
\ENDWHILE
\RETURN \textbf{no-cluster};
\end{algorithmic}
\end{algorithm}

\vspace{-0.3cm}

The last step of our algorithm takes the clustering  $C'$ returned by
\emph{Expand-Landmarks} and improves it.  We compute a set $L'$ that contains
exactly one landmark from each cluster $C'_{i} \in C'$ (any landmark is sufficient), and assign each point $x \in S$ to the cluster corresponding to the closest landmark in $L'$.

We now present our main theoretical guarantee for Algorithm~\ref{alg-main}.

\begin{theorem}\label{thm:Main}
Given a metric space $M = (X,d)$, where $d$ is unknown, and a set of points $S$,
if the instance $(S,d)$ satisfies the $(1+\alpha,\epsilon)$-property for the $k$-median objective function
and if each cluster in the target clustering $C_{T}$ has size at least
$(4+51/\alpha)\epsilon n$, then Landmark-Clustering outputs a
clustering that is $\epsilon$-close to $C_{T}$ with probability  $1-\delta$ in time $O((k + \ln \frac{1}{\delta})|S| \log |S|)$ using $O(k + \ln \frac{1}{\delta})$ \emph{one versus all} distance queries.
\end{theorem}

Before we prove the theorem, we will introduce some notation and use an analysis similar to the one in \citep{bbg} to argue about the structure of the clustering instance.
Let $\epsilon^{\ast} = \dist(C_{T},C^{\ast})$.  By our assumption that the
$k$-median clustering of $S$ satisfies the $(1+\alpha,\epsilon)$-property we have
$\epsilon^{\ast} < \epsilon$.  Since each cluster in the target clustering has at
least $(4+51/\alpha)\epsilon n$ points, and the \emph{optimal k-median
clustering} $C^{\ast}$ differs from the target clustering by $\epsilon^{\ast}n
\le \epsilon n$ points, each cluster in $C^{\ast}$ must have at least
$(3+51/\alpha)\epsilon n$ points.

Let us define the \emph{critical distance} $d_{\mathrm{crit}} = \frac{\alpha w}{17
\epsilon}$.  We call a point $x$ \emph{good} if both $w(x) < d_{\mathrm{crit}}$ and
$w_{2}(x) - w(x) \ge 17 d_{\mathrm{crit}}$, else $x$ is called \emph{bad}. In other words, the
$good$ points are those points that are close to their own cluster center and far
from any other cluster center.  In addition, we will break up the \emph{good}
points into \emph{good sets} $X_{i}$, where $X_{i}$ is the set of the
\emph{good} points in the optimal cluster $C_{i}^{\ast}$.
So each set $X_{i}$ is the ``core'' of the optimal cluster $C_{i}^{\ast}$.

Note that the distance between two points $x,y \in X_{i}$ satisfies $d(x,y) \le
d(x,c^{\ast}_{i}) + d(c^{\ast}_{i},y) = w(x) + w(y) < 2 d_{\mathrm{crit}}$.  In addition, the distance between any two points in different good sets is greater than $16d_{\mathrm{crit}}$.  To see this, consider a pair of points $x \in X_{i}$ and $y \in X_{j \ne i}$.  The distance from $x$ to $y$'s cluster center $c^{\ast}_{j}$ is at least $17 d_{\mathrm{crit}}$. By the triangle inequality, $d(x,y) \ge d(x,c^{\ast}_{j}) - d(y,c^{\ast}_{j}) > 17 d_{\mathrm{crit}} - d_{\mathrm{crit}} = 16 d_{\mathrm{crit}}$.

If the k-median instance $(M,S)$ satisfies the
$(1+\alpha,\epsilon)$-property with respect to $C_{T}$, and each cluster in
$C_{T}$ has size at least $2\epsilon n$, then

\begin{enumerate}
\item less than $(\epsilon - \epsilon^{\ast})n$ points $x \in S$ on which $C_{T}$ and $C^{\ast}$ agree have $w_{2}(x) - w(x) < \frac{\alpha w}{\epsilon}$.
\item at most $17 \epsilon n / \alpha$ points $x \in S$ have $w(x) \ge \frac{\alpha w}{17 \epsilon}$.
\end{enumerate}

The first part is proved by \citet{bbg}.  The intuition is that if too many points on which $C_{T}$ and $C^{\ast}$ agree
are close enough to the second-closest center among $\lbrace
c_{1}^{\ast},c_{2}^{\ast},\ldots,c_{k}^{\ast} \rbrace$, then we can move them to
the clusters corresponding to those centers, producing a clustering that is far
from $C_{T}$, but whose objective value is close to OPT, violating the
$(1+\alpha,\epsilon$)-property.  The second part follows from the fact that
$\sum_{x \in S} w(x) = OPT = wn.$

Then using these facts and the definition of $\epsilon^{\ast}$ it follows that at
most $\epsilon^{\ast} n + (\epsilon - \epsilon^{\ast})n + 17 \epsilon n / \alpha
= \epsilon n + 17 \epsilon n / \alpha = (1 + 17/ \alpha) \epsilon n = b$ points
are bad.  Hence each $\vert X_{i} \vert = \vert C^{\ast}_{i} \backslash B \vert
\ge (2+34/\alpha)\epsilon n = 2b$.


In the remainder of this section we prove that given this structure of the clustering instance, \emph{Landmark-Clustering} finds an accurate clustering. We first show that almost surely the set of landmarks returned by \emph{Landmark-Selection} has the property that each of the cluster cores has a landmark near it.  We then argue that given a set of landmarks with this property, \emph{Expand-Landmarks} finds a partition $C'$ that clusters most of the points in each core correctly.  We conclude with the proof of the theorem, which argues that the clustering returned by the last step of our procedure is a further improved clustering that is very close to $C^{\ast}$ and $C_{T}$.

The \emph{Landmark-Clustering} algorithm first uses \emph{Landmark-Selection}($q,\iter$) to choose a set of landmark points.  The following lemma proves that for an appropriate choice of $q$ after selecting only $\iter = O(k + \ln \frac{1}{\delta})$ landmarks with probability at least $1 - \delta$ one of them is closer than $2 d_{\mathrm{crit}}$ to some point in each good set.

\begin{lemma}\label{lemma:PointCoverage}
Given $L$ = Landmark-Selection $(2b,4k+16\ln\frac{1}{\delta})$, with probability at least $1-\delta$ there is a landmark closer than $2 d_{crit}$ to some point in each good set.
\end{lemma}

\begin{proof}
Because there are at most $b$ bad points and in each iteration we uniformly at random choose one of $2b$ points, the probability that a good point is added to $L$ is at least 1/2 in each iteration. Using a Chernoff bound we show that the
probability that fewer than $k$ good points have been added to $L$ after $t > 2k$ iterations is less than $e^{-t(1-\frac{2k}{t})^{2}/4}$ (Lemma~\ref{lemma:FewGoodPoints}). For $t = 4k + 16 \ln \frac{1}{\delta}$ \begin{displaymath}
e^{-t(1-\frac{2k}{t})^{2}/4} < e^{-(4k + 16 \ln \frac{1}{\delta})0.5^2/4} < e^{-16 \ln \frac{1}{\delta}/16} = \delta.
\end{displaymath}
Therefore after $t = 4k + 16 \ln \frac{1}{\delta}$ iterations this probability is smaller than $\delta$.

We argue that once we select $k$ good points using our procedure, one of them must be closer than $2 d_{\mathrm{crit}}$ to some point in each good set.  Note that the selected good points must be distinct because we must have chosen at least $k$ good points after $b+k$ iterations and we cannot choose the same point twice in the first $n-2b$ iterations.  There are two possibilities regarding the first $k$ good points added to $L$: they are either selected from distinct good sets, or at least two of them are selected from the same good set.

If the former is true then the statement trivially holds.  If the latter is true, consider the first time that a second point is chosen from the
same good set $X_{i}$.  Let us call these two points $x$ and $y$, and assume that
$y$ is chosen after $x$.  The distance between $x$ and $y$ must be less than $2 d_{\mathrm{crit}}$ because they are in the same good set.
Therefore when $y$ is chosen, $\min_{l \in L} d(l,y) \le d(x,y) < 2
d_{\mathrm{crit}}$.  Moreover, $y$ is chosen from $\lbrace s_{n - 2b +
1},...,s_{n} \rbrace$, where $\min_{l \in L} d(l,s_{i}) \le \min_{l \in L}
d(l,s_{i+1})$.  Therefore when $y$ is chosen, at least $n-2b+1$ points $s \in S$
(including $y$) satisfy $\min_{l \in L} d(l,s) \le \min_{l \in L} d(l,y) < 2
d_{\mathrm{crit}}$.  Since each good set satisfies $ \vert X_{i} \vert \ge 2b$, it follows that there must be a landmark closer than $2 d_{crit}$ to some point in each good set.

\end{proof}

\begin{lemma}\label{lemma:FewGoodPoints}
The probability that fewer than $k$ good points have been chosen as landmarks
after $t > 2k$ iterations of Landmark-Selection is less than
$e^{-t(1-\frac{2k}{t})^{2}/4}$.
\end{lemma}

\begin{proof}
Let $X_{i}$ be an indicator random variable defined as follows: $X_{i} = 1$
if point chosen in iteration $i$ is a good point, and 0 otherwise.  Let $X=\sum_{i=1}^{t}X_{i}$, and $\mu$ be the expectation of $X$.  In other words, $X$ is the number of good points chosen after $t$ iterations of the algorithm, and $\mu$ is its expected value.

Because in each round we uniformly at random choose one of $2b$ points and there are at most $b$ bad points in total, E$\lbrack X_{i} \rbrack \ge 1/2$ and hence $\mu \ge t/2$.  By the Chernoff bound, for any $\delta > 0$, Pr$\lbrack X < (1-\delta)\mu \rbrack < e^{-\mu\delta^{2}/2}$.

If we set $\delta = 1 - \frac{2k}{t}$, we have $(1-\delta)\mu = (1-(1-\frac{2k}{t}))\mu \ge (1-(1-\frac{2k}{t}))t/2 = k$.  Assuming that $t \ge 2k$, it follows that $\textrm{Pr} \lbrack X < k \rbrack \le \textrm{Pr} \lbrack X < (1-\delta)\mu \rbrack < e^{-\mu\delta^{2}/2} = e^{-\mu(1-\frac{2k}{t})^{2}/2} \le e^{-t/2(1-\frac{2k}{t})^{2}/2}$.
\end{proof}

The algorithm then uses the \emph{Expand-Landmarks} procedure to find a $k$-clustering $C'$.  The following lemma states that $C'$ is an accurate clustering, and has an additional property that is relevant for the last part of the algorithm.

\begin{lemma}
Given a set of landmarks $L$ chosen by Landmark-Selection
so that the condition in Lemma~\ref{lemma:PointCoverage} is satisfied,
Expand-Landmarks$(b+1,n-b,L)$ returns a $k$-clustering $C' = \lbrace C'_{1},C'_{2},
\ldots C'_{k} \rbrace$ in which each cluster contains points from a distinct good
set $X_{i}$.  If we let $\sigma$ be a bijection mapping each good set $X_{i}$ to
the cluster $C'_{\sigma(i)}$ containing points from $X_{i}$, the distance between
$c^{\ast}_{i}$ and any landmark $l$ in $C'_{\sigma(i)}$ satisfies
$d(c^{\ast}_{i},l) < 5 d_{\mathrm{crit}}$.
\end{lemma}

\begin{proof}
Lemma~\ref{lemma:NoOverlap} argues that since the good sets $X_{i}$ are well-separated, for $r < 4
d_{\mathrm{crit}}$ no ball of radius $r$ can overlap more than one $X_{i}$, and
two balls that overlap different $X_{i}$ cannot share any points.  Moreover,
since we only consider balls that have more than $b$ points in them, and the
number of bad points is at most $b$, each ball in $G_{B}$ must overlap some good
set.  Lemma~\ref{lemma:BallContainsGoodSet} argues that since there is a landmark near each good set, there
is a value of $r^{\ast} < 4 d_{\mathrm{crit}}$ such that each $X_{i}$ is
contained in some ball around a landmark of radius $r^{\ast}$.  We can use these
facts to argue for the correctness of the algorithm.

First we observe that for $r = r^{\ast}$, $G_{B} $ has exactly $k$ components and
each good set $X_{i}$ is contained within a distinct component.  Each ball in
$G_{B}$ overlaps with some $X_{i}$, and by Lemma~\ref{lemma:NoOverlap}, since $r^{\ast} < 4
d_{\mathrm{crit}}$, we know that each ball in $G_{B}$ overlaps with exactly one
$X_{i}$.  From Lemma~\ref{lemma:NoOverlap} we also know that balls that overlap different $X_{i}$
cannot share any points and are thus not connected in $G_{B}$.  Therefore balls
that overlap different $X_{i}$ will be in different components in $G_{B}$.
Moreover, by Lemma~\ref{lemma:BallContainsGoodSet} each $X_{i}$ is contained in some ball of radius
$r^{\ast}$.  For each good set $X_{i}$ let us designate by $B_{i}$ a ball that
contains all the points in $X_{i}$ (Figure~\ref{fig:figure2}), which is in $G_{B}$ since the
size of each good set satisfies $\vert X_{i} \vert > b$.  Any ball in $G_{B}$
that overlaps $X_{i}$ will be connected to $B_{i}$, and will thus be in the same
component as $B_{i}$.  Therefore for $r = r^{\ast}$, $G_{B}$ has exactly $k$
components, one for each good set $X_{i}$ that contains all the points in
$X_{i}$.

\begin{figure}
\begin{center}
\includegraphics[width=70mm,height=50mm]{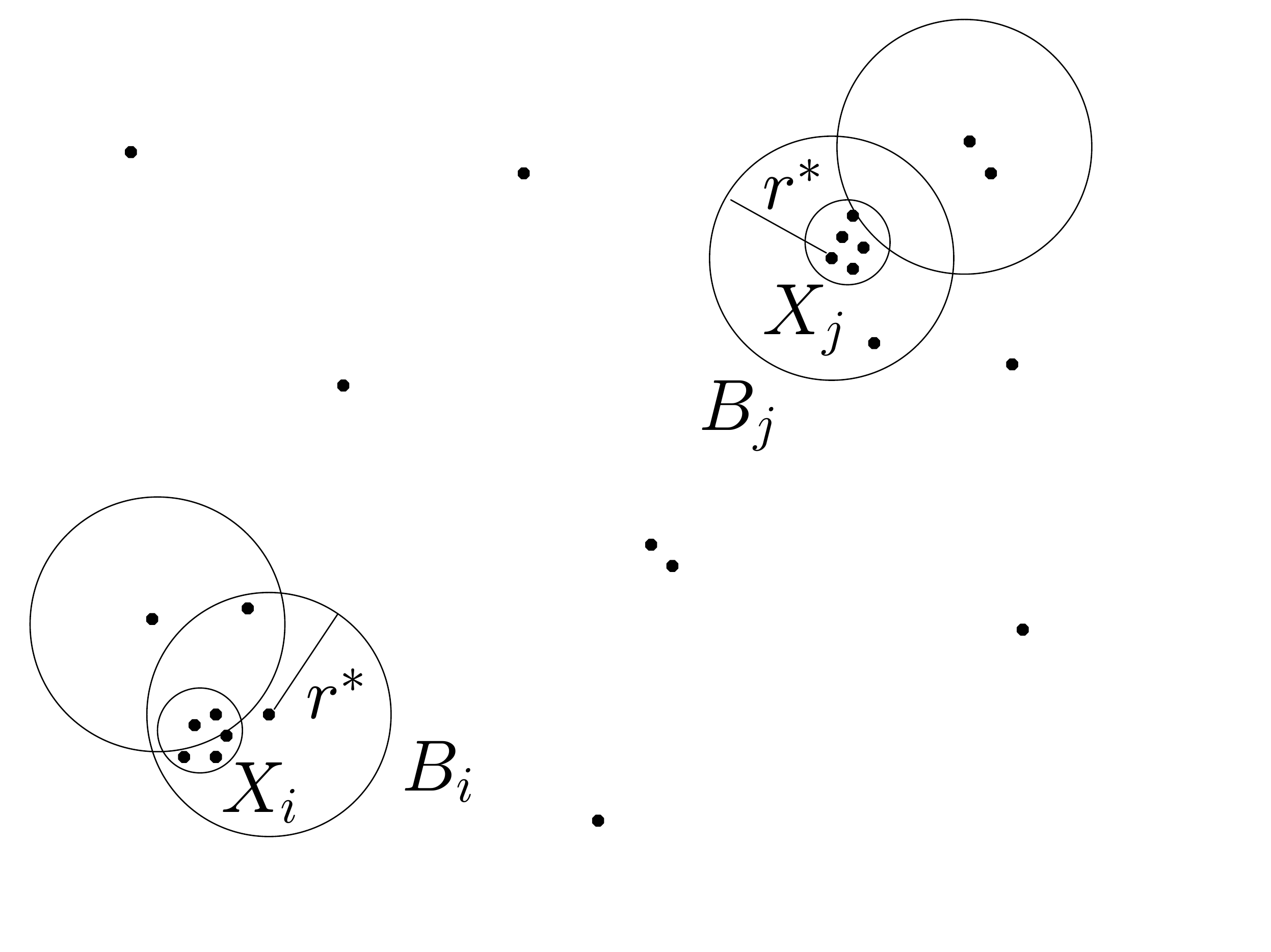}
\caption{Balls $B_{i}$ and $B_{j}$ of radius $r^{\ast}$ are shown, which contain good sets $X_{i}$ and $X_{j}$, respectively.  The radius of the balls is small in comparison to the distance between the good sets. \label{fig:figure2}}
\end{center}
\end{figure}

Since there are at least $n-b$ good points that are in some $X_{i}$, this means
that for $r = r^{\ast}$ the number of points that are in some ball in $G_{B}$
(which are in $\Clustered$) is at least $n-b$.  Hence the condition in line 4 of \emph{Expand-Landmarks}
will be satisfied and the algorithm will terminate and return a $k$-clustering in
which each cluster contains points from a distinct good set $X_{i}$.

Now let us suppose that we start with $r=0$.  Consider the first value of $r=r'$
for which the condition in line 4 is satisfied.  At this point $G_{B}$ has
exactly $k$ components and the number of points that are not in these components
is at most $b$.  It must be the case that $r' \le r^{\ast} < 4 d_{\mathrm{crit}}$
because we know that the condition is satisfied for $r = r^{\ast}$, and we are
considering all relevant values of $r$ in ascending order.  As before, each ball
in $G_{B}$ must overlap some good set $X_{i}$.  Again using Lemma~\ref{lemma:NoOverlap} we argue
that since $r < 4 d_{\mathrm{crit}}$, no ball can overlap more than one $X_{i}$
and two balls that overlap different $X_{i}$ cannot share any points.  It follows
that each component of $G_{B}$ contains points from a single $X_{i}$ (so we
cannot merge the good sets).  Moreover, since the size of each good set satisfies
$\vert X_{i} \vert > b$, and there are at most $b$ points left out of $G_{B}$,
each component must contain points from a distinct $X_{i}$ (so we cannot split
the good sets).  Thus we will return a $k$-clustering in which each cluster
contains points from a distinct good set $X_{i}$.

To prove the second part of the statement, let $\sigma$ be a bijection matching
each good set $X_{i}$ to the cluster $C'_{\sigma(i)}$ containing points from
$X_{i}$.  Clearly, $C'_{\sigma(i)}$ is made up of points in balls of radius $r <
4 d_{\mathrm{crit}}$ that overlap $X_{i}$.  Consider any such ball $B_{l}$ around
landmark $l$ and let $s^{\ast}$ denote any point on which $B_{l}$ and $X_{i}$
overlap.  By the triangle inequality, the distance between $c_{i}^{\ast}$ and $l$ satisfies
$d(c_{i}^{\ast},l) \le d(c_{i}^{\ast},s^{\ast}) + d(s^{\ast},l) < d_{\mathrm{crit}} + r < 5 d_{\mathrm{crit}}.$  Therefore the distance between $c_{i}^{\ast}$ and any landmark $l \in
C'_{\sigma(i)}$ satisfies $d(c_{i}^{\ast},l) < 5 d_{\mathrm{crit}}$.
\end{proof}

\begin{lemma}\label{lemma:NoOverlap}
A ball of radius $r < 4 d_{\mathrm{crit}}$ cannot contain points from more than
one good set $X_{i}$, and two balls of radius $r < 4 d_{\mathrm{crit}}$ that
overlap different $X_{i}$ cannot share any points.
\end{lemma}

\begin{proof}
To prove the first part, consider a ball $B_{l}$ of radius $r < 4
d_{\mathrm{crit}}$ around landmark $l$.  In other words, $B_{l} = \lbrace s \in S
\mid d(s,l) \le r \rbrace$.  If $B_{l}$ overlaps more than one good set, then it
must have at least two points from different good sets $x \in X_{i}$ and $y \in
X_{j}$.  By the triangle inequality it follows that $d(x,y) \le d(x,l) + d(l,y)
\le 2r < 8 d_{\mathrm{crit}}$.  However, we know that $d(x,y) > 16
d_{\mathrm{crit}}$, giving a contradiction.

To prove the second part, consider two balls $B_{l_{1}}$ and $B_{l_{2}}$ of
radius $r < 4 d_{\mathrm{crit}}$ around landmarks $l_{1}$ and $l_{2}$.  In other
words, $B_{l_{1}} = \lbrace s \in S \mid d(s,l_{1}) \le r \rbrace$, and
$B_{l_{2}} = \lbrace s \in S \mid d(s,l_{2}) \le r \rbrace$.  Assume that
they overlap with different good sets $X_{i}$ and $X_{j}$: $B_{l_{1}} \cap X_{i}
\ne \emptyset$ and $B_{l_{2}} \cap X_{j} \ne \emptyset$.  For the purpose of
contradiction, let's assume that $B_{l_{1}}$ and $B_{l_{2}}$ share at least one
point: $B_{l_{1}} \cap B_{l_{2}} \ne \emptyset$, and use $s^{\ast}$ to refer to
this point.  By the triangle inequality, it follows that the distance between any
point $x \in B_{l_{1}}$ and $y \in B_{l_{2}}$ satisfies $d(x,y) \le d(x,s^{\ast})
+ d(s^{\ast},y) \le \lbrack d(x,l_{1}) + d(l_{1},s^{\ast}) \rbrack + \lbrack
d(s^{\ast},l_{2}) + d(l_{2},y) \rbrack \le 4r < 16 d_{\mathrm{crit}}.$

Since $B_{l_{1}}$ overlaps with $X_{i}$ and $B_{l_{2}}$ overlaps with $X_{j}$, it
follows that there is a pair of points $x \in X_{i}$ and $y \in X_{j}$ such that
$d(x,y) < 16 d_{\mathrm{crit}}$, a contradiction.  Therefore if
$B_{l_{1}}$ and $B_{l_{2}}$ overlap different good sets, $B_{l_{1}} \cap
B_{l_{2}} = \emptyset$.
\end{proof}

\begin{lemma}\label{lemma:BallContainsGoodSet}
Given a set of landmarks $L$ chosen by Landmark-Selection
so that the condition in Lemma~\ref{lemma:PointCoverage} is satisfied, there is some value of $r^{\ast} < 4 d_{\mathrm{crit}}$ such
that each $X_{i}$ is contained in some ball $B_{l}$ around landmark $l \in L$ of
radius $r^{\ast}$.
\end{lemma}

\begin{proof}
For each good set $X_{i}$ choose a point $s_{i} \in X_{i}$ and a landmark $l_{i}
\in L$ that satisfy $d(s_{i},l_{i}) < 2 d_{\mathrm{crit}}$.  The distance between $l_{i}$ and each point
$x \in X_{i}$ satisfies $d(l_{i},x) \le d(l_{i},s_{i}) + d(s_{i},x) < 2
d_{\mathrm{crit}} + 2 d_{\mathrm{crit}} = 4 d_{\mathrm{crit}}.$

Consider $r^{\ast} = \textrm{max}_{l_{i}} \textrm{max}_{x \in X_{i}} d(l_{i},x)$.
Clearly, each $X_{i}$ is contained in a ball $B_{l_{i}}$ of radius $r^{\ast}$
and $r^{\ast} < 4 d_{\mathrm{crit}}$.
\end{proof}

\begin{lemma}\label{lemma:CloseToOwnLandmark}
Suppose the distance between $c^{\ast}_{i}$ and any landmark $l$ in
$C'_{\sigma(i)}$ satisfies $d(c^{\ast}_{i},l) < 5 d_{\mathrm{crit}}$.  Then given
point $x \in C^{\ast}_{i}$ that satisfies $w_{2}(x) - w(x) \ge 17
d_{\mathrm{crit}}$, for any $l_{1} \in C'_{\sigma(i)}$ and $l_{2} \in
C'_{\sigma(j \ne i)}$ it must be the case that $d(x,l_{1}) < d(x,l_{2})$.
\end{lemma}

\begin{proof}
We will show that $d(x,l_{1}) < w(x) + 5d_{\mathrm{crit}}$ \textbf{(1)}, and $d(x,l_{2}) > w(x) + 12d_{\mathrm{crit}}$ \textbf{(2)}.  This implies that $d(x,l_{1}) < d(x,l_{2})$.

To prove \textbf{(1)}, by the triangle inequality $d(x,l_{1}) \le d(x,c^{\ast}_{i}) + d(c^{\ast}_{i},l_{1}) = w(x) + d(c^{\ast}_{i},l_{1}) < w(x) + 5 d_{\mathrm{crit}}$.  To prove \textbf{(2)}, by the triangle inequality $d(x,c_{j}^{\ast}) \le d(x,l_{2}) + d(l_{2},c^{\ast}_{j})$.  It follows that $d(x,l_{2}) \ge d(x,c_{j}^{\ast}) - d(l_{2},c^{\ast}_{j})$.  Since $d(x,c_{j}^{\ast}) \ge w_{2}(x)$ and $d(l_{2},c^{\ast}_{j}) < 5 d_{\mathrm{crit}}$ we have
\begin{equation}
d(x,l_{2}) > w_{2}(x) - 5 d_{\mathrm{crit}}.
\end{equation}
Moreover, since $w_{2}(x) - w(x) \ge 17 d_{\mathrm{crit}}$ we have
\begin{equation}
w_{2}(x) \ge 17 d_{\mathrm{crit}} + w(x).
\end{equation}
Combining Equations 1 and 2 it follows that $d(x,l_{2}) > 17 d_{cri} + w(x) - 5 d_{\mathrm{crit}} = w(x) + 12 d_{\mathrm{crit}}$.
\end{proof}

\begin{proof}[Theorem~\ref{thm:Main}]
After using Landmark-Selection to choose $O(k + \ln \frac{1}{\delta})$ points, with probability at least $1-\delta$ there is a landmark closer than $2 d_{\mathrm{crit}}$ to some point in each good set.  Given a set of landmarks with this property, each cluster in the clustering $C' = \lbrace C'_{1},C'_{2}, \ldots C'_{k} \rbrace$ output by \emph{Expand-Landmarks} contains points from a distinct good set $X_{i}$. This clustering can exclude up to $b$ points, all of which may be good.  Nonetheless, this means that $C'$ may disagree with $C^{\ast}$ on only the bad points and at most $b$ good points.  The number of points that $C'$ and $C^{\ast}$ disagree on is therefore at most $2b = O(\epsilon n /\alpha)$.  Thus, $C'$ is at least $O(\epsilon/\alpha)$-close to $C^{\ast}$, and at least $O(\epsilon/\alpha + \epsilon)$-close to $C_{T}$.

Moreover, $C'$ has an additional property that allows us to find a clustering
that is $\epsilon$-close to $C_{T}$.  If we use $\sigma$ to denote a bijection
mapping each good set $X_{i}$ to the cluster $C'_{\sigma(i)}$ containing points
from $X_{i}$, any landmark $l \in C'_{\sigma(i)}$ is closer than $5 d_{\mathrm{crit}}$ to $c^{\ast}_{i}$.   We can use this observation to find all points that satisfy one of the properties of the
good points: points $x$ such that $w_{2}(x) - w(x) \ge 17 d_{\mathrm{crit}}$.
Let us call these points the \emph{detectable} points.  To clarify, the
detectable points are those points that are much closer to their own
cluster center than to any other cluster center in $C^{\ast}$, and the
\emph{good} points are a subset of the detectable points that are also
very close to their own cluster center.

To find the detectable points using $C'$, we choose some landmark $l_{i}$
from each $C'_{i}$.  For each point $x \in S$, we then insert $x$ into the
cluster $C''_{j}$ for $j$ = argmin$_{i}d(x,l_{i})$.  Lemma~\ref{lemma:CloseToOwnLandmark} argues that each detectable point in $C^{\ast}_{i}$ is closer to every landmark in $C'_{\sigma(i)}$ than to any landmark in $C'_{\sigma(j \ne i)}$.  It follows that $C''$ and $C^{\ast}$ agree on all the detectable points.  Since there are fewer than $(\epsilon - \epsilon^{\ast})n$ points on which $C_{T}$ and $C^{\ast}$ agree that are not detectable, it follows that $\dist(C'',C_{T}) < (\epsilon - \epsilon^{\ast}) + \dist(C_{T},C^{\ast}) = (\epsilon - \epsilon^{\ast}) + \epsilon^{\ast} = \epsilon$.

Therefore using $O(k + \ln \frac{1}{\delta})$ landmarks we get an accurate clustering with probability at least $1-\delta$.  The runtime of \emph{Landmark-Selection} is $O(\vert L \vert n)$, where $\vert L \vert$ is the number of landmarks.  Using a min-heap to store all landmark-point pairs and a disjoint-set data structure to keep track of the connected components of $G_{B}$, \emph{Expand-Landmarks} can be implemented in $O(\vert L \vert n \log n)$ time.  A detailed description of this implementation is given in the next section.  The last part of our procedure takes $O(kn)$ time, so the runtime of our implementation is $O(\vert L \vert n \log n)$.  Therefore to get an accurate clustering with probability $1-\delta$ the runtime of our algorithm is $O((k + \ln \frac{1}{\delta}) n \log n)$.  Moreover, we only consider the distances between the landmarks and other points, so we only use $O(k + \ln \frac{1}{\delta})$ one versus all distance queries.
\end{proof}

\section{Implementation of Expand-Landmarks}

In order to efficiently expand balls around landmarks, we build a min-heap $H$ of landmark-point pairs $(l,s)$, where the key of each pair is the distance between $l$ and $s$.  In each iteration we find $(l^{\ast},s^{\ast})$ = $H$.deleteMin(), and then add $s^{\ast}$ to items($l^{\ast}$), which stores the points in $B_{l^{\ast}}$.  We store points that have been clustered (points in balls of size larger than $s_{\min}$) in the set Clustered.

Our implementation assigns each clustered point $s$ to a ``representative'' landmark, denoted by $l(s$).  The representative landmark of $s$ is the landmark $l$ of the first large ball $B_{l}$ that contains $s$.  To efficiently update the components of $G_{B}$, we maintain a disjoint-set data structure $U$ that contains sets corresponding to the connected components of $G_{B}$, where each ball $B_{l}$ is represented by landmark $l$.  In other words, $U$ contains a set $\lbrace l_{1},l_{2},\ldots,l_{i} \rbrace$ iff $B_{l_{1}}, B_{l_{2}}, \ldots,B_{l_{i}}$ form a connected component in $G_{B}$.  For each large ball $B_{l}$ our algorithm considers all points $s \in B_{l}$ and performs Update-Components($l,s$), which works as follows.  If $s$ does not have a representative landmark we assign it to $l$, otherwise $s$ must already be in $B_{l(s)}$, and we assign $B_{l}$ to the same component as $B_{l(s)}$.  If none of the points in $B_{l}$ are assigned to other landmarks, it will be in its own component.  A detailed description of the algorithm is given below.

\vspace{-0.6cm}

\begin{algorithm}[H]
\caption{Expand-Landmarks($s_{\min},n',L$)}
\begin{algorithmic}[1]
\STATE A = ();
\FOR{each $s \in S$}
\STATE $l(s)$ = null;
\FOR{each $l \in L$}
\STATE A.add$((l,s),d(l,s))$;
\ENDFOR
\ENDFOR
\STATE $H$ = build-heap($A$);
\FOR{each $l \in L$}
\STATE items($l$) = ();
\ENDFOR
\STATE Set Clustered = ();
\STATE U = ();
\WHILE{$H$.hasNext()}
\STATE  $(l^{\ast},s^{\ast})$ = $H$.deleteMin();
\STATE items($l^{\ast}$).add($s^{\ast}$);
\IF{items($l^{\ast}$).size() == $s_{\min}$}
\STATE Activate($l^{\ast}$);
\ENDIF
\IF{items($l^{\ast}$).size() $> s_{\min}$}
\STATE Update-Components($l^{\ast},s^{\ast}$);
\ENDIF
\IF{Clustered.size() $\ge n'$ and $U$.size() == $k$}
\RETURN Format-Clustering();
\ENDIF
\ENDWHILE
\RETURN \textbf{no-cluster};
\end{algorithmic}
\end{algorithm}

\vspace{-0.6cm}

\begin{algorithm}[H]
\caption{Update-Components($l,s$)}
\begin{algorithmic}[1]
\IF{$l(s)$ == null}
\STATE $l(s)$ = $l$;
\ELSE
\STATE $c_{1}$ = $U$.find($l$);
\STATE $c_{2}$ = $U$.find($l(s)$);
\STATE $U$.union($c_{1},c_{2}$);
\ENDIF
\end{algorithmic}
\end{algorithm}

\vspace{-0.6cm}

\begin{algorithm}[H]
\caption{Activate($l$)}
\begin{algorithmic}[1]
\STATE $U$.MakeSet($l$);
\FOR{each $s \in$ items($l$)}
\STATE Update-Components($l,s$);
\STATE Clustered.add($s$);
\ENDFOR
\end{algorithmic}
\end{algorithm}


\begin{algorithm}[H]
\caption{Format-Clustering()}
\begin{algorithmic}[1]
\STATE C = ();
\FOR{each Set $L$ in $U$}
\STATE Set Cluster = ();
\FOR{each $l \in L$}
\FOR{each $s \in$ items($l$)}
\STATE Cluster.add($s$);
\ENDFOR
\ENDFOR
\STATE $C$.add(Cluster);
\ENDFOR
\RETURN $C$;
\end{algorithmic}
\end{algorithm}


During the execution of the algorithm the connected components of $G_{B}$ correspond to the sets of $U$ (where each ball $B_{l}$ is represented by landmark $l$).  Suppose that $B_{l_{1}}$ and $B_{l_{2}}$ are connected in $G_{B}$, then $B_{l_{1}}$ and $B_{l_{2}}$ must overlap on some point $s$.  Without loss of generality, suppose $s$ is added to $B_{l_{1}}$ before it is added to $B_{l_{2}}$.  When $s$ is added to $B_{l_{1}}$, $l(s)$ = $l_{1}$ if $s$ does not yet have a representative landmark (lines 1-2 of Update-Components), or $l(s)$ = $l'$ and both $l_{1}$ and $l'$ are put in the same set (lines 4-6 of Update-Components).  When $s$ is added to $B_{l_{2}}$, if $l(s)$ = $l_{1}$, then $l_{1}$ and $l_{2}$ will be put in the same set.  If $l(s)$ = $l'$, $l'$ and $l_{2}$ will be put in the same set, which also contains $l_{1}$.

It follows that whenever $B_{l_{1}}$ and $B_{l_{2}}$ are in the same connected component in $G_{B}$, $l_{1}$ and $l_{2}$ will be in the same set in $U$.  Moreover, if $B_{l_{1}}$ and $B_{l_{2}}$ are not in the same component in $G_{B}$, then $l_{1}$ and $l_{2}$ can never be in the same set in $U$  because both start in distinct sets (line 1 of Activate), and it is not possible for a set containing $l_{1}$ to be merged with a set containing $l_{2}$.

It takes $O(\vert L \vert n)$ time to build $H$ (linear in the size of the heap).  Each deleteMin() operation takes
$O(\log(\vert L \vert n))$ (logarithmic in the size of the heap), which is
equivalent to $O(\log(n))$ because $\vert L \vert \le n$.  If $U$ is implemented by a union-find algorithm
Update-Components takes amortized time of $O(\alpha(|L|)$, where $\alpha$ denotes the inverse
Ackermann function.  Moreover, Update-Components may only be called once for each iteration of the while loop in
Expand-Landmarks (it is either called immediately on $l^{\ast}$ and $s^{\ast}$ if
$B_{l^{\ast}}$ is large enough, or it is called when the ball grows large enough
in Activate).  All other operations also take time proportional to the number of
landmark-point pairs.  So the runtime of this algorithm is $O(\vert L \vert n)$ +
$\iter \cdot O(\log n + \alpha(\vert L \vert))$, where $\iter$ is the number of iterations of the while
loop.  As the number of iterations is bounded by $|L|n$, and $\alpha(\vert L \vert)$ is effectively constant,
this gives a worst-case running time of $O(\vert L
\vert n \log n)$.

\section{Empirical Study}

We use our \emph{Landmark Clustering} algorithm to cluster proteins using
sequence similarity.  As mentioned in the Introduction, one versus all distance
queries are particularly relevant in this setting because of sequence database
search programs such as BLAST \citep{blast}  (Basic Local Alignment Search Tool).  BLAST aligns the queried sequence to sequences in the database, and produces a
``bit score'' for each alignment, which is a measure of its quality (we invert
the bit score to make it a distance).  However, BLAST does not consider
alignments with some of the sequences in the database, in which case we assign
distances of infinity to the corresponding sequences.  We observe that if we
define distances in this manner they almost form a metric in practice: when we draw triplets of sequences at random and check the distances between them the triangle inequality is almost always satisfied.  Moreover, BLAST is very successful at detecting sequence homology in large sequence databases, therefore it is
plausible that clustering using these distances satisfies the $(c,\epsilon)$-property
for some relevant clustering $C_{T}$.

We perform experiments on datasets obtained from two classification databases: Pfam \citep{pfam}, version 24.0,
October 2009; and SCOP \citep{scop}, version 1.75, June 2009.  Both
of these sources classify proteins by their evolutionary
relatedness, therefore we can use their classifications as a ground truth to
evaluate the clusterings produced by our algorithm and other methods.


Pfam classifies proteins using hidden Markov models (HMMs) that represent multiple sequence alignments. There are two levels in
the Pfam classification hierarchy: family and clan.  In our clustering
experiments we compare with a classification at the family level because the
relationships at the clan level are less likely to be discerned with sequence
alignment.  In each experiment we randomly select several large families (of size
between 1000 and 10000) from Pfam-A (the manually curated part of the classification), retrieve the sequences of the proteins in these
families, and use our \emph{Landmark-Clustering} algorithm to cluster the
dataset.


SCOP groups proteins on the basis of their 3D structures, so it only classifies
proteins whose structure is known.  Thus the datasets from SCOP are much smaller
in size.  The SCOP classification is also hierarchical: proteins are grouped by
class, fold, superfamily, and family.  We consider the classification at the
superfamily level because this seems most appropriate given that we are only using sequence information.
As with the Pfam data, in each experiment we create a dataset by randomly
choosing several superfamilies (of size between 20 and 200), retrieve the
sequences of the corresponding proteins, and use our \emph{Landmark-Clustering}
algorithm to cluster the dataset.

Once we cluster a particular dataset, we compare the clustering to the manual classification using the distance measure from the theoretical part of our work.  To find the fraction of misclassified points under the optimal matching of clusters in $C$ to clusters in $C'$ we solve a minimum weight bipartite matching problem where the cost of matching $C_{i}$ to $C'_{\sigma(i)}$ is $\vert C_{i} - C'_{\sigma(i)} \vert / n$.  In addition, we compare clusterings to manual classifications using the F-measure, which was used in another study of clustering protein sequences \citep{spectralClusteringProteinSeqs}.  The F-measure gives a score between 0 and 1, where 1 indicates
an exact match between the two clusterings (see Appendix A).
The F-measure has also been used in other studies \citep[see][]{fMeasureUse}, and is related to our notion of distance (Lemma~\ref{lemma:clusteringDistanceRelationship} in Appendix A).

\subsection{Choice of Parameters}

To run \textit{Landmark-Clustering}, we set $k$ using the number of clusters in the ground truth clustering.  For each Pfam dataset we use $40k$ landmarks/queries, and for each SCOP dataset we use $30k$ landmarks/queries.  In addition, our algorithm uses three parameters $(q,s_{\min},n')$ whose value is set in the proof based on  $\alpha$ and $\epsilon$, assuming that the clustering instance satisfies the $(1+\alpha,\epsilon)$-property.  In
practice we must choose some value for each parameter.  In our experiments we set them as a function of the number of points in the dataset, and the number of clusters.  We set $q = 2n/k$, $s_{\min} = 0.05n/k$ for Pfam datasets, and $s_{min} = 0.1n/k$ for SCOP datasets, and $n' = 0.5n$.  Since the selection of landmarks is randomized, for each dataset we perform several clusterings, compare each to the ground truth, and report the median quality.

\emph{Landmark-Clustering} is most sensitive to the $s_{\min}$ parameter, and will not report a clustering if $s_{\min}$ is too small or too large.  We recommend trying several reasonable values of $s_{\min}$, in increasing or decreasing order, until you get a clustering and none of the clusters are too large.  If you get a clustering where one of the clusters is very large, this likely means that several ground truth clusters have been merged.  This may happen because $s_{\min}$ is too small causing balls of outliers to connect different cluster cores, or $s_{\min}$ is too large causing balls in different cluster cores to overlap.

The algorithm is less sensitive to the $n'$ parameter.  However, if you set $n'$ too large some ground truth clusters may be merged, so we recommend using a smaller value ($0.5n \le n' \le 0.7n$) because all of the points are still clustered during the last step.  Again, for some values of $n'$ the algorithm may not output a clustering, or output a clustering where some of the clusters are too large.  Our algorithm is least sensitive to the $q$ parameter.  Using more landmarks (if you can afford it) can make up for a poor choice of $q$.

\subsection{Results}


\begin{figure}
  \centering
  \subfloat[Comparison using fraction of misclassified points]{\label{fig:clusteringComparisonA}\includegraphics[width=80mm,height=60mm]{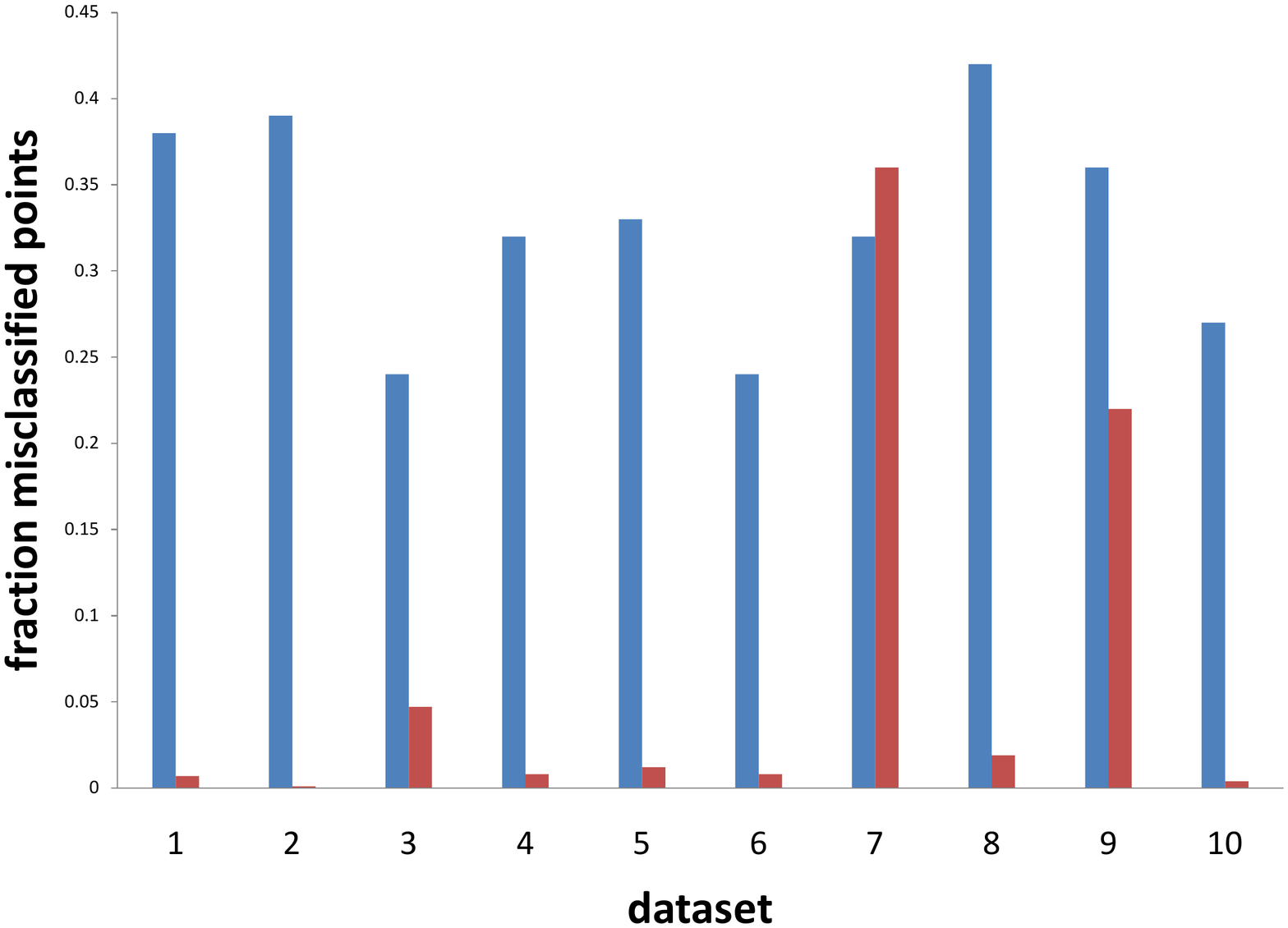}}
  \subfloat[Comparison using the F-measure]{\label{fig:clusteringComparisonB}\includegraphics[width=80mm,height=60mm]{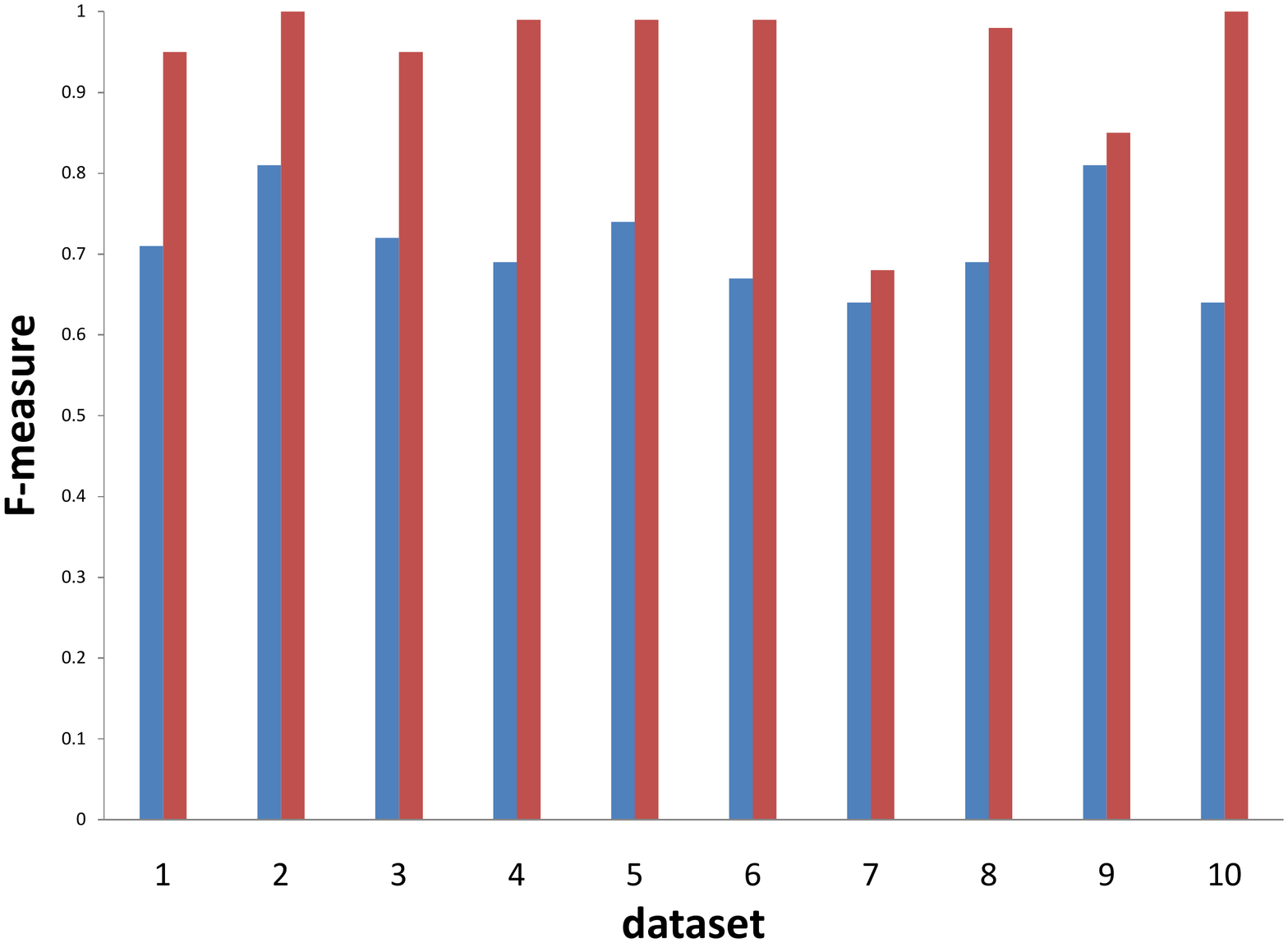}}
  \caption{Comparing the performance of $k$-means in the embedded space (blue) and \emph{Landmark-Clustering} (red) on 10 datasets from Pfam.  Datasets \textbf{1-10} are created by randomly choosing 8 families from Pfam of size $s$, $1000 \le s \le 10000$. \textbf{(a)} Comparison using the distance measure from the theoretical part of our work.  \textbf{(b)} Comparison using the F-measure.}
  \label{fig:figure3}
\end{figure}

Figure~\ref{fig:figure3} shows the results of our experiments on the Pfam
datasets.  One can see that for most of the datasets (other than datasets 7 and
9) we find a clustering that is almost identical to the ground truth.  These
datasets are very large, so as a benchmark for comparison we can only consider
algorithms that use a comparable amount of distance information (since we do not
have the full distance matrix).  A natural choice is the following algorithm: randomly choose a set of landmarks $L$, $\vert L \vert = d$; embed each point in a $d$-dimensional space using distances to $L$; use $k$-means clustering in this space (with distances given by the Euclidian norm).  Our embedding scheme is a Lipschitz embedding with singleton subsets \citep[see][]{virtualLandmarks}, which gives distances with low distortion for points near each other in a metric space.

Notice that this procedure uses exactly $d$ one versus all distance queries, so
we can set $d$ equal to the number of queries used by our algorithm.  We expect
this algorithm to work well, and if you look at Figure~\ref{fig:figure3} you can see that it finds reasonable clusterings.  Still, the clusterings reported by this algorithm do not closely match the Pfam classification, showing that our results are indeed significant.

Figure~\ref{fig:figure4} shows the results of our experiments on the SCOP datasets. These results are not as good, which is likely because the SCOP
classification at the superfamily level is based on biochemical and
structural evidence in addition to sequence evidence.
By contrast, the Pfam classification is based entirely on sequence information.
Still, because the SCOP datasets are much smaller, we can compare our algorithm
to methods that require distances between all the points.  In particular,
\citet{spectralClusteringProteinSeqs} showed that spectral clustering using sequence data works well when applied to the proteins in SCOP.  Thus we use the exact method described by \citet{spectralClusteringProteinSeqs} as a benchmark for comparison on the SCOP datasets.  Moreover, other than clustering
randomly generated datasets from SCOP, we also consider the two main examples
from Paccanaro et al., which are labeled \textbf{A} and
\textbf{B} in the figure.  From Figure~\ref{fig:figure4} we can see that the performance of
\emph{Landmark-Clustering} is comparable to that of the spectral method, which is
very good considering that the algorithm used by \citet{spectralClusteringProteinSeqs} significantly
outperforms other clustering algorithms on this data.  Moreover, the spectral clustering algorithm requires the full distance matrix as input, and takes much longer to run.


\begin{figure}
  \centering
  \subfloat[Comparison using fraction of misclassified points]{\label{fig:clusteringComparisonA}\includegraphics[width=80mm,height=60mm]{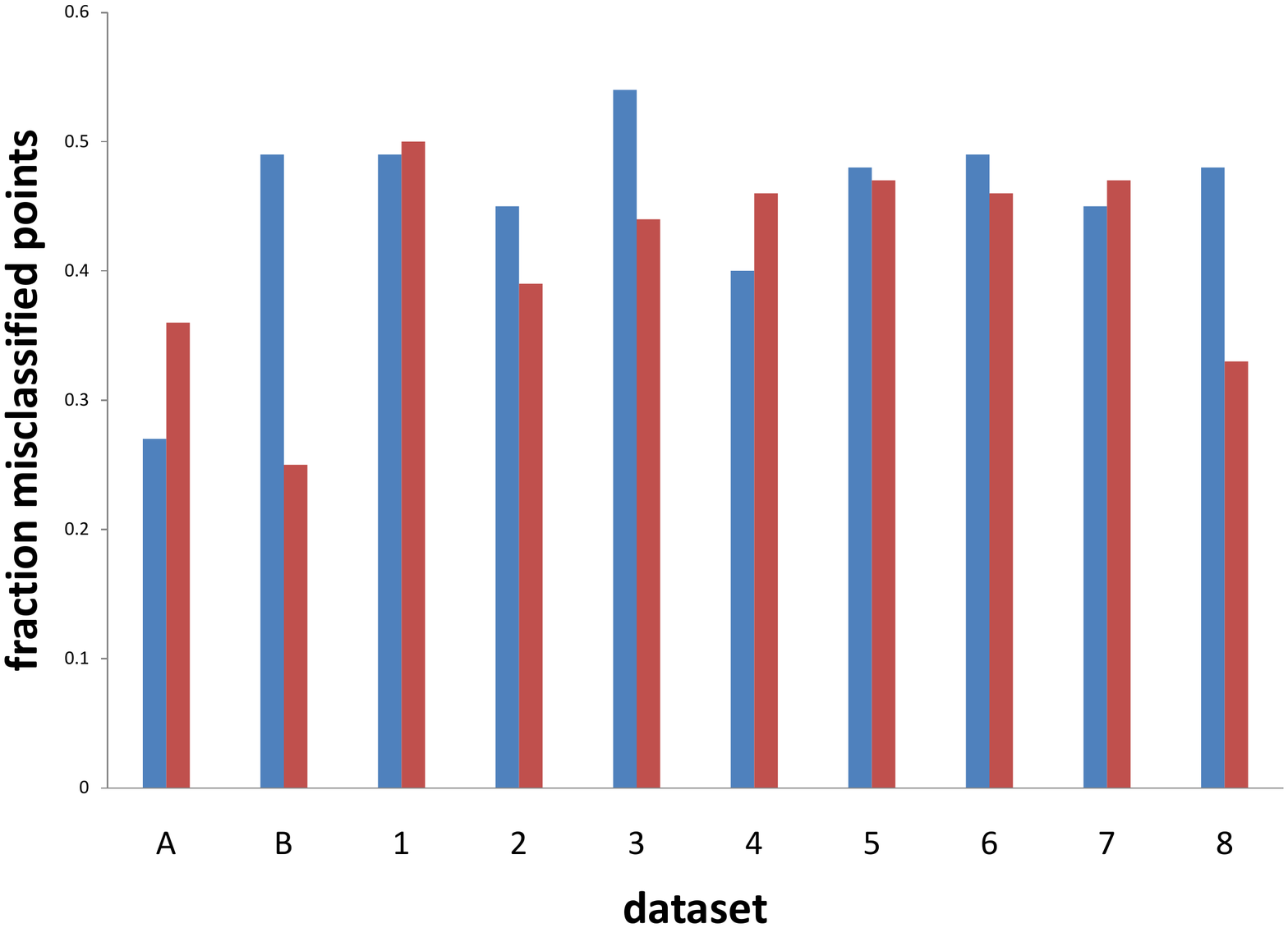}}
  \subfloat[Comparison using the F-measure]{\label{fig:clusteringComparisonB}\includegraphics[width=80mm,height=60mm]{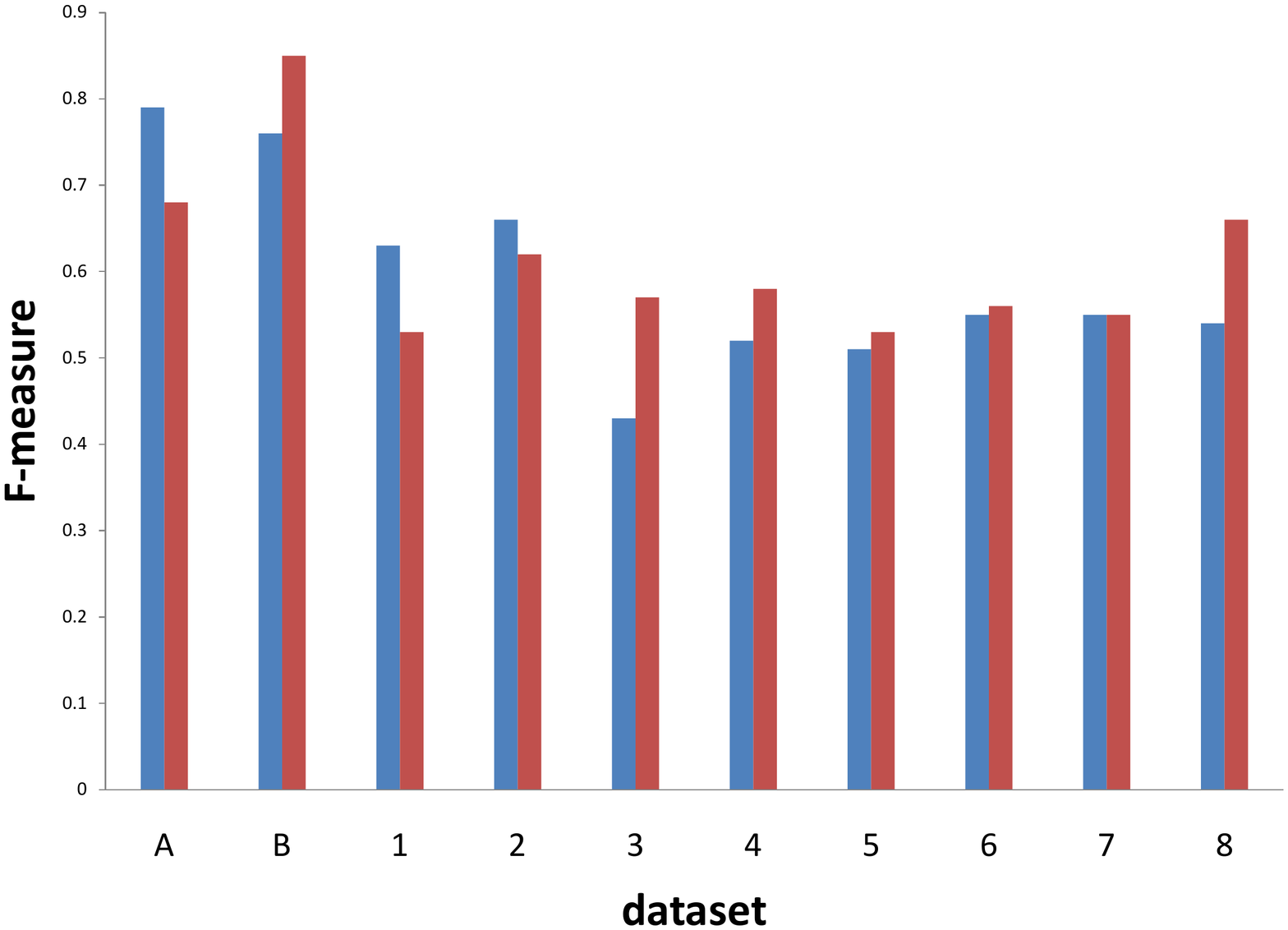}}
  \caption{Comparing the performance of spectral clustering (blue) and \emph{Landmark-Clustering} (red) on 10 datasets from SCOP.  Datasets \textbf{A} and \textbf{B} are the two main examples from \cite{spectralClusteringProteinSeqs}, the other datasets (\textbf{1-8}) are created by randomly choosing 8 superfamilies from SCOP of size $s$, $20 \le s \le 200$. \textbf{(a)} Comparison using the distance measure from the theoretical part of our work.  \textbf{(b)} Comparison using the F-measure.}
  \label{fig:figure4}
\end{figure}

\subsection{Testing the $(c,\epsilon)$ property}

To see whether the $(c,\epsilon)$ property is a reasonable assumption for our data, we look at whether our datasets have the structure implied by our assumption.  We do this by measuring the separation of the ground truth clusters in our datasets.  For each dataset in our study, we sample some points from each ground truth cluster.  We consider whether the sampled points are more similar to points in the same cluster than to points in other clusters.  More specifically, for each point we record the median within-cluster similarity, and the maximum between-cluster similarity.  If our datasets indeed have well-separated cluster cores, as implied by our assumption, then for a lot of the points the median within-cluster similarity should be significantly larger than the maximum between-cluster similarity.  We can see that this is indeed the case for the Pfam datasets.  However, this is not typically the case for the SCOP datasets, where most points have little similarity to the majority of the points in their ground truth cluster.  These observations explain our results on the two sets of data: we are able to accurately cluster the Pfam datasets, and our algorithm is much less accurate on the SCOP datasets.  The complete results of these experiments can be found at {http://cs-people.bu.edu/kvodski/clusteringProperties/description.html}.

\section{Conclusion and Open Questions}

In this work we presented a new algorithm for clustering large datasets with
limited distance information. As opposed to previous settings, our goal was not
to approximate some objective function like the $k$-median objective, but to find
clusterings close to the ground truth. We proved that our algorithm yields
accurate clusterings with only a small number of one versus all distance queries,
given a natural assumption about the structure of the clustering instance. This
assumption has been previously analyzed by \citet{bbg}, but in the full distance
information setting. By contrast, our algorithm uses only a small number of
queries, it is much faster, and it has effectively the same formal performance
guarantees as the one introduced by \citet{bbg}.

To demonstrate the practical use of our algorithm, we clustered protein sequences
using a sequence database search program as the one versus all query.  We
compared our results to gold standard manual classifications of protein
evolutionary relatedness given in Pfam \citep{pfam} and SCOP \citep{scop}. We find
that our clusterings are comparable in accuracy to the classification given in
Pfam. For SCOP our clusterings are as accurate as state of the art methods, which
take longer to run and require the full distance matrix as input.

Our main theoretical guarantee assumes large target clusters. It would be interesting to
design a provably correct algorithm for the case of small clusters as well.


\acks{Konstantin Voevodski was supported by an IGERT Fellowship through NSF grant DGE-0221680 awarded to the ACES Training Program at BU Center for Computational Science.  Maria Florina Balcan was supported in part by NSF grant CCF-0953192,
by ONR grant N00014-09-1-0751 and AFOSR grant FA9550-09-1-0538.
Heiko R{\"o}glin was supported by a Veni grant from the Netherlands Organisation
for Scientific Research.  Shang-Hua Teng was supported in part by NSF grant CCR-0635102.}


\newpage

\appendix
\section*{Appendix A.}
In this section we reproduce the definition of F-measure, which is another way to evaluate the distance between two clusterings.  We also show a relationship between our measure of distance and the F-measure.
\subsection*{A.1 F-measure}

The F-measure compares two clusterings $C$ and $C'$ by matching
each cluster in $C$ to a cluster in $C'$ using a harmonic mean of
Precision and Recall, and then computing a ``per-point'' average.  If we match
$C_{i}$ to $C'_{j}$, Precision is defined as $P(C_{i},C'_{j}) =
\frac{\vert C_{i} \cap C'_{j} \vert}{\vert C_{i} \vert}.$  Recall
is defined as $R(C_{i},C'_{j}) = \frac{\vert C_{i} \cap C'_{j}
\vert}{\vert C_{j} \vert}.$  For $C_{i}$ and $C'_{j}$ the harmonic mean of
Precision and Recall is then equivalent to $\frac{2 \cdot \vert C_{i}
\cap C'_{j} \vert}{\vert C_{i} \vert + \vert C'_{j} \vert}$, which we denote by $\pr(C_{i},C'_{j})$ to simplify notation.  The F-measure is then defined as
\begin{displaymath}
F(C,C') = \frac{1}{n}
\sum_{C_{i} \in C} \vert C_{i} \vert \max_{C'_{j} \in C'} \pr(C_{i},C'_{j}).
\end{displaymath}
Note that this quantity is between 0 and 1, where 1 corresponds to an
exact match between the two clusterings.

\begin{lemma}\label{lemma:clusteringDistanceRelationship}
Given two clusterings $C$ and $C'$, if $\dist(C,C') = d$ then F($C,C') \ge 1 - 3d/2$.
\end{lemma}

\begin{proof}
Denote by $\sigma$ the optimal matching of clusters in $C$ to clusters in $C'$, which achieves a misclassification of $dn$ points.  We show that just considering $\pr(C_{i},C'_{\sigma(i)})$ for each $C_{i} \in C$ achieves an F-measure of at least $1-3d/2$:
\begin{displaymath}
F(C,C') \ge \frac{1}{n} \sum_{C_{i} \in C} \vert C_{i} \vert \pr(C_{i},C'_{\sigma(i)}) \ge 1 - 3d/2.
\end{displaymath}

To see this, for a match of $C_{i}$ to $C'_{\sigma(i)}$ we denote by $m^{1}_{i}$ the number of points that are in $C_{i}$ but not in $C'_{\sigma(i)}$, and by $m^{2}_{i}$ the number of points that are in $C'_{\sigma(i)}$ but not in $C_{i}$: $m^{1}_{i} = \vert C_{i} - C'_{\sigma(i)} \vert$, $m^{2}_{i} = \vert  C'_{\sigma(i)} - C_{i} \vert$.  Because the total number of misclassified points is $dn$ it follows that
\begin{displaymath}
\sum_{C_{i} \in C} m^{1}_{i} = \sum_{C_{i} \in C} m^{2}_{i} = dn.
\end{displaymath}

By definition, $\vert C_{i} \cap C'_{\sigma(i)} \vert = \vert C_{i} \vert - m^{1}_{i}$.  Moreover, $\vert C'_{\sigma(i)}
\vert = \vert C'_{\sigma(i)} \cap C_{i} \vert + m^{2}_{i} \le \vert C_{i} \vert + m^{2}_{i}$.  It follows that
\begin{displaymath}
\pr(C_{i},C'_{\sigma(i)}) = \frac{2(\vert C_{i} \vert - m^{1}_{i})}{\vert C_{i} \vert + \vert C'_{\sigma(i)} \vert} \ge  \frac{2(\vert C_{i} \vert - m^{1}_{i})}{2 \vert C_{i} \vert + m^{2}_{i}} = \frac{2 \vert C_{i} \vert + m^{2}_{i}}{2 \vert C_{i} \vert + m^{2}_{i}} - \frac{m^{2}_{i} + 2m^{1}_{i}}{2 \vert C_{i} \vert + m^{2}_{i}} \ge 1 - \frac{m^{2}_{i} + 2m^{1}_{i}}{2 \vert C_{i} \vert}.
\end{displaymath}

We can now see that

\begin{displaymath}
\frac{1}{n} \sum_{C_{i} \in C} \vert C_{i} \vert \pr(C_{i},C'_{\sigma(i)}) \ge \frac{1}{n} \sum_{C_{i} \in C} \vert C_{i} \vert (1 - \frac{m^{2}_{i} + 2m^{1}_{i}}{2 \vert C_{i} \vert}) = \frac{1}{n} \sum_{C_{i} \in C} \vert C_{i} \vert  - \frac{1}{2n} \sum_{C_{i} \in C} m^{2}_{i} + 2m^{1}_{i} = 1 - \frac{3dn}{2n}.
\end{displaymath}

\end{proof}

\bibliography{references}

\end{document}